\documentclass[12pt]{article}

\usepackage{amssymb}
\usepackage{amsmath}
\usepackage{amsthm}
\usepackage{enumitem}
\usepackage{tabularx}
\usepackage{graphicx}
\usepackage[ruled]{algorithm2e}
\usepackage[margin=0.5in]{geometry}
\usepackage{xr}

\title{Utilitarians Without Utilities: \\ Maximizing Social Welfare for Graph Problems \\ using only Ordinal Preferences \\Full Version}

\author{
	Ben Abramowitz\\Department of Computer Science\\Rensselaer Polytechnic Institute\\abramb@rpi.edu
	\and
	Elliot Anshelevich\\Department of Computer Science\\Rensselaer Polytechnic Institute\\eanshel@cs.rpi.edu
}

\begin{document}
\maketitle

\newtheorem{theorem}{Theorem}
\newtheorem{corollary}{Corollary}
\newtheorem{lemma}{Lemma}
\newtheorem{definition}{Definition}
\newtheorem{claim}{Claim}
\newtheorem{observation}{Observation}
\newtheorem{example}{Example}

\begin{abstract}
We consider ordinal approximation algorithms for a broad class of utility maximization problems for multi-agent systems. In these problems, agents have utilities for connecting to each other, and the goal is to compute a maximum-utility solution subject to a set of constraints. We represent these as a class of graph optimization problems, including matching, spanning tree problems, TSP, maximum weight planar subgraph, and many others. We study these problems in the ordinal setting: latent numerical utilities exist, but we only have access to ordinal preference information, i.e., every agent specifies an ordering over the other agents by preference. We prove that for the large class of graph problems we identify, ordinal information is enough to compute solutions which are close to optimal, thus demonstrating there is no need to know the underlying numerical utilities. For example, for problems in this class with bounded degree $b$ a simple ordinal greedy algorithm always produces a ($b+1$)-approximation; we also quantify how the quality of ordinal approximation depends on the sparsity of the resulting graphs. In particular, our results imply that ordinal information is enough to obtain a 2-approximation for Maximum Spanning Tree; a 4-approximation for Max Weight Planar Subgraph; a 2-approximation for Max-TSP; and a 2-approximation for various Matching problems.
\end{abstract}

%____________________________________________________%
%____________________________________________________%
%____________________________________________________%

\section{Introduction}
Human beings are terrible at expressing their feelings quantitatively. For example, when forming collaborations people may be able to order their peers from ``best to collaborate with" to worst, but would have a difficult time assigning exact numeric values to the acuteness of these preferences.
In other words, even when numerical (possibly latent) utilities exist, in many settings it is much more reasonable to assume that we only know {\em ordinal} preferences: every agent specifies the {\em order} of their preferences over the alternatives, instead of a numerical value for each alternative. Recently there has been a lot of work using such an {\em implicit utilitarian approach}, especially for matching and social choice (see Related Work), in situations where obtaining true numerical utilities may be difficult. 
Amazingly, as this line of work shows, it is often possible to design algorithms and mechanisms which perform well using only ordinal information. 
In fact, ordinal algorithms often perform almost as well as {\em omniscient} mechanisms which know the true underlying numerical utilities, instead of just the ordinal preferences induced by these utilities.

In this work we consider a relatively general network formation setting. All problems considered herein are modeled by an undirected complete graph $G = (\mathcal{N,E})$, where the (symmetric) weight $w(x,y)$ of each edge $(x,y) \in \mathcal{E}$ represents the hidden utility of connecting agents $x,y \in \mathcal{N}$. The goal is to form a maximum-weight (i.e., maximum utility) graph which obeys some given constraints. For example, constraints may include bounds on the maximum degree, on component size, and many others. This framework includes such problems as matching, group formation, TSP, and many others as special cases; see Section \ref{preliminaries} for example constraints and how they lead to different important settings.

In many settings we may not have access to the true edge weights $w(x,y)$. Instead, each agent $x \in \mathcal{N}$ reports a strict \emph{preference ordering} over the other agents $\mathcal{N} - \{x\}$ with whom it can connect. We assume this ordering to be consistent with the latent weights, so that $w(x,y) > w(x,z)$ implies that $x$ prefers $y$ to $z$. While it is clearly impossible to form an optimal (i.e., maximum-utility) solution without direct knowledge of the edge weights, we show how to design good approximation algorithms for selecting a maximum weight subgraph of $G$, subject to a large set of constraints. As usual in this line of work, the measure of performance is simply the sum of the agent utilities. Our paper provides good approximations for a broad class of ordinal analogues to graph optimization problems representing utility maximization for multi-agent systems. Note that unlike {\em all} previous work mentioned here, we do not make additional assumptions about the structure of the edge weights: we {\em do not} assume either that the agent utilities are normalized (as in \cite{boutilier2015optimal,brandt2016handbook,caragiannis2017subset,caragiannis2011voting,christodoulou2016social}), nor that they form a metric space (as in \cite{anshelevich2015approximating,anshelevich2017randomized,anshelevich2016blind,anshelevich2016truthful,caragiannis2016truthful,goel2017metric,gross2017vote,skowron2017social}). Thus, our results demonstrate how well one can perform using only ordinal information without additional assumptions.

\paragraph{ABC Systems} More specifically, we define a class of constraints called \emph{ABC Systems} which consists of three types of constraints. The first two are the familiar constraints which bound the maximum degree of each node and the maximum component size, or number of nodes in a connected component. The third constraint is a much more general requirement called attachment which only applies to nodes that are already in the same connected component. The maximization problem for an ABC System is to compute a maximum weight subgraph $S \subseteq \mathcal{E}$ of an undirected, complete graph $G = (\mathcal{N,E})$ such that in $S$ every node has degree at most $b$, every connected component has size (number of nodes) at most $c$, and $S \in \mathcal{A}$ for some arbitrary \textbf{attachment set} $\mathcal{A}$. A collection of subgraphs $\mathcal{A}$ is an attachment set of $G$ if the following properties hold for all subgraphs $F \subseteq \mathcal{E}$:\\
1) \emph{Heredity:} If $F \in \mathcal{A}$ and $F' \subseteq F$ then $F' \in \mathcal{A}$.\\
2) \emph{Attachment:} If $F \in \mathcal{A}$ and $F + e \notin \mathcal{A}$ for some $e = (u,v) \notin F$, then there is a $(u,v)$-path in $F$. \footnote{We use the ``+" and ``-" notation when adding or removing a single edge to or from a set.}

Note that all three of our constraints possess the heredity property which enables greedy heuristics, like the Ordinal Greedy algorithm we introduce in Section \ref{sub:algorithm}, to construct valid solutions. The intuition behind the attachment property is that if $F \in \mathcal{A}$ but $F + e \notin \mathcal{A}$, then $e$ must have both endpoints in the same component. Therefore, the number of such edges whose addition would violate $\mathcal{A}$ within any component of size $x$ is bounded by $\frac{x \cdot \min{(b, x-1)}}{2} - (x-1)$, where $x$ is bounded by $c$.

The utility maximization objective for ABC Systems encompasses a wide variety of well-known problems central to algorithm design. The examples we address in this paper include Max Weight $b$-Matching, Maximum Weight Spanning Tree, Maximum Traveling Salesperson, and Max Weight Planar Subgraph. Our results also encapsulate many other interesting optimization problems for ABC Systems which we will not discuss directly, like finding the maximum weight subgraph with minimum girth $k$, maximum cycle length $l$, or which excludes a variety of graph minors (including all 2-edge-connected minors).  As we show, all such problems are amenable to knowing only secondhand ordinal information, instead of the true numerical utilities.

\subsection{Our Contributions}

\begin{table*}[t]
 \begin{tabular}{|c | c | c | c|}
 \hline
Maximization Problem & Ordinal Greedy & Omniscient Greedy & Best Known \\ [0.5ex]
 \hline \hline
  ABC System& $b+1$ & $b+1$ & - \\
 \hline
  AB System& $\max\{2,d+1\}$ & $\max\{2, d+1\}$ & - \\
 \hline
 Spanning Tree & 2 & 1 & 1 \\
 \hline
  Planar Subgraph & 4 & 3 & 72/25 \cite{calinescu2003new}\\
 \hline
 Traveling Salesperson & 2 & 2 & 9/7 \cite{paluch20097}\\
 \hline
 $b$-Matching & 2 & 2 & 1 \\
 %\hline
 %Matching & 2 & 2 & 1 \\
 \hline

\end{tabular}
\caption{Here we compare our results for Ordinal Greedy, known results for Omniscient Greedy, and the best known polynomial-time algorithm with full-information. All of our bounds are tight except for the one on Planar Subgraph.}
\label{table}
\end{table*}

Most algorithmic techniques for maximizing utility for the full-information setting do not translate to the ordinal information setting. These typically rely on non-local information, like comparisons between weights of non-adjacent edges, or comparing the total weights of sets of edges. This is not possible using only ordinal information. Even the fundamental, and well-studied \cite{vince2002framework,edmonds1971matroids,rado1942theorem}, \emph{Omniscient Greedy} algorithm, which adds edges in strictly non-increasing order of their weight, cannot be executed using only ordinal information. Instead, we focus on the natural {\em Ordinal Greedy} algorithm (defined in Section \ref{sub:algorithm}), which adds edges iteratively as long as the edge $(x,y)$ being added is the most preferred edge for both $x$ and $y$ out of all the possible edges which could be added at that time. Ordinal Greedy has some very nice properties: in addition to being natural and providing high-utility solutions (as we prove in this paper), it also always creates {\em pairwise stable} solutions: no pair of agents would have incentive to destroy some of their links and form a new link connecting them. Note that the performance of Ordinal Greedy can be very different from Omniscient Greedy: see Example \ref{example:MST} in Section \ref{sub:algorithm} for intuition of why this must be.

In this paper, we analyze the performance of Ordinal Greedy for many ABC Systems (see Table \ref{table}). We first prove that for general ABC Systems, Ordinal Greedy always produces a solution with weight at most factor $b+1$ away from optimum, and that this factor is tight. In other words, for general problems including all those in Table \ref{table}, as long as the number of connections for each node must be bounded by some small $b$, then using only ordinal information it is possible to compete with the best possible solution, and thus with any algorithm which knows the true numerical utilities. Such results tell us that when $b$ is small, there is no need to find out the hidden edge weights/utilities; knowing the ordinal preferences is good enough.

Second, we show that by relaxing the component size constraint (i.e., setting $c$ to be unbounded) we can achieve significant improvements. For convenience, we call ABC Systems with the component size constraint relaxed {\em AB Systems}. We prove that as long as any solution formed by such an AB System is guaranteed to be at most {\em $d$-sparse}, then Ordinal Greedy forms a solution within a factor of $d+1$ from optimum. Since the sparsity of a graph is at most half the average degree in any subgraph, we know that $d \leq \frac{b}{2}$. Therefore an AB System is at worst $(\frac{b}{2}+1)$-approximable, giving us a factor of 2 improvement over general ABC Systems.

This result is more powerful than it may first appear, as many important constraints yield sparse solutions. For example, since all tours and trees are 1-sparse, Ordinal Greedy provides a 2-approximation for both Maximum Traveling Salesperson and Maximum Weight Spanning Tree. And since all planar graphs are 3-sparse, we obtain a 4-approximation for Max Weight Planar Subgraph which uses only ordinal information.

Lastly, we consider Max Weight $b$-Matching, in which the only constraint is that each agent can be matched with at most $b$ others. This is simply an ABC System with unbounded $c$ and $\mathcal{A}$ being all possible sets. We prove that our approximation factor drops to a constant 2, regardless of the value of $b$.

To prove the ordinal approximations above, we first demonstrate that for any ABC System (and any graphic system with the heredity property defined earlier) Ordinal Greedy achieves its worst approximation on an instance with weight function $w: \mathcal{E} \rightarrow \{0,1\}$. We use this fact heavily to establish our approximation bounds, and believe it to be of independent interest. Note that similar results for Omniscient Greedy have relied critically on the fact that it selects edges in strictly non-increasing order by weight. Clearly this does not and cannot hold for the ordinal setting, as it is even possible for the minimum weight edge of the graph to be selected before the maximum weight edge. Because of this, our proofs require completely new approaches and techniques.

%____________________________________________________%
%____________________________________________________%
%____________________________________________________%
\subsection{Related Work}
Historically, it has been common to approach problems in the ordinal setting with a normative view by designing mechanisms which satisfy axiomatic properties, like stability or truthfulness. These axiomatic properties are useful in many applications, but do not provide a quantitative measure of the quality of a solution. The notion of distortion and the implicit utilitarian framework were first introduced by 
\cite{procaccia2006distortion} in the context of voting to provide such a measure. Since then the distortion, or approximation factor, of various ordinal utility maximization mechanisms has been studied, particularly for matchings \cite{anshelevich2016blind,anshelevich2016truthful,anshelevich2017tradeoffs,christodoulou2016social,caragiannis2016truthful} and social choice \cite{anshelevich2015approximating,anshelevich2017randomized,boutilier2015optimal,brandt2016handbook,caragiannis2017subset,caragiannis2011voting,goel2017metric,gross2017vote,skowron2017social}.

Our work is unlike that in social choice, since we consider network formation problems where agent preferences are expressed over one another. In the context of matchings, \cite{anshelevich2016blind,anshelevich2016truthful} develop various matching algorithms as a black-box to provide approximations for a variety of matching and clustering problems under the implicit utilitarian view. 
Additionally, \cite{christodoulou2016social} and \cite{caragiannis2016truthful} provide results for one-sided matchings and \cite{anshelevich2017tradeoffs} consider bipartite matchings. However, {\em all previous work} on approximation for utility maximization mentioned above either assumes the underlying weights form a metric space \cite{anshelevich2015approximating,anshelevich2017randomized,anshelevich2016blind,anshelevich2016truthful,anshelevich2017tradeoffs,goel2017metric,gross2017vote,skowron2017social,caragiannis2016truthful} or are normalized \cite{boutilier2015optimal,brandt2016handbook,caragiannis2017subset,caragiannis2011voting,christodoulou2016social}, with only two exceptions. The first exception is that Maximum Traveling Salesperson yields a 2-approximation without the metric assumption \cite{anshelevich2016truthful}. We prove this result as part of a much more general theorem using much more general techniques. The second is the result developed in the full-information setting by \cite{preis1999linear} and affirmed in the ordinal setting by \cite{anshelevich2016blind}, that Max Weight Matching yields a 2-approximation without any assumptions on the weights. We generalize this result to all $b$-Matchings instead of only $b=1$.

The work most similar to ours is \cite{anshelevich2016blind} which bounds the distortion of ordinal mechanisms for several problems, including Maximum Traveling Salesperson, but relies heavily on the assumption that the weights obey the triangle inequality. For the Ordinal Greedy algorithm, \cite{anshelevich2016blind} show the metric assumption implies that any two edges which are both most preferred by their respective endpoints at some iteration must be within a factor of 2 of one another, even if they are not adjacent. By contrast, this non-local information is unavailable to us in our model. Our paper is unique in that we identify a large class of problems for which assumptions on the weights are unnecessary to achieve good approximations to optimum with only ordinal information.

The Omniscient Greedy algorithm has been studied extensively. In fact, it is known to be optimal on exactly the set of independence systems (any system with the heredity property) which are matroids \cite{edmonds1971matroids,rado1942theorem}, which includes Maximum Spanning Trees. \cite{korte1978analysis} showed that Omniscient Greedy provides good approximations for many independence systems, including matching and symmetric TSP. 
\cite{dyer1985analysis} further demonstrated that Omniscient Greedy provides a tight 3-approximation for Max Weight Planar Subgraph. These results were later reformulated as k-extendibility by \cite{mestre2006greedy}, who applies this idea to a diverse set of problems, including $b$-Matching. Unfortunately, the proofs for all results just mentioned rely critically on Omniscient Greedy selecting edges in strictly non-increasing order, making them untenable in the ordinal setting. No ordinal algorithm can yield optimum solutions, even for matroids. However, as we show in Table \ref{table}, our results compete well with the best known polynomial-time algorithms for ABC Systems.

%_______________________________%
\section{Model and Problem Statements} \label{preliminaries}

The input for all problems in this paper is a set $\mathcal{N}$ of agents (nodes) of size $n$, and a strict preference ordering for each $x \in \mathcal{N}$ over the edges adjacent to $x$. The preference orderings reported by each agent are induced by a set of hidden symmetric weights $w(x,y) = w(y,x)$ for all $x, y \in \mathcal{N}$. The set of hidden weights corresponds to an undirected, complete graph $G = (\mathcal{N,E})$ with non-negative weight function $w : \mathcal{E} \rightarrow \mathbb{R}^+$. The transitive relation of the individual preference orderings for all agents determines a partial ordering $\sigma$ over all edges; note that some pairs of edges may end up being incomparable in $\sigma$ (see Example 1). The preference ordering $\sigma$ is said to be \textbf{consistent} with the hidden weights if $\forall x, y, z \in \mathcal{N}$, if $x$ prefers $y$ to $z$ then it must be that $w(x,y) \geq w(x,z)$. If an edge $e_1$ is known to be at least as large as edge $e_2$ according to this partial ordering $\sigma$, then we will say that $e_1$ {\em dominates} $e_2$ in $\sigma$. The problems we consider are optimization problems where for an instance $(w, \sigma)$ the objective is to compute the subgraph of $G$ with maximum total edge weight, subject to a set of constraints, knowing only $\sigma$.

For weight function $w$, we let $OPT_w$ be the optimal solution for the weights prescribed by $w$, and we let $w(OPT_w)$ be the total weight of the optimal solution evaluated by $w$. Likewise, we use $S$ to denote our constructed solution and $w(S)$ its weight. Our approximation factor for a problem is therefore $\alpha = \max\limits_{(w,\sigma)} \frac{w(OPT_w)}{w(S)}$, where $S$ is any solution returned by our algorithm for $(w,\sigma)$.

%_______________________________%
Recall the definition of ABC Systems. Given constraints $\mathcal{A}, b, c$ we can say without loss of generality that $b \leq c-1$, because if any node has $c$ or more neighbors in a component, this component would have to be of size greater than $c$. When $b = c - 1$ this effectively removes the node degree constraint. Similarly, when $c = n$, the component size is effectively unbounded. Therefore, $1 \leq b < c \leq n$. Likewise, when $\mathcal{A}$ = all subgraphs of $G$, this effectively annuls the attachment set constraint. Some specific problems which we consider in this paper are as follows.

\begin{description} \label{problems}
\item[Max ABC:] $c \leq n$, $b < c$, $\mathcal{A}$ = any attachment set of $G$

\item[Max AB:] $c = n$, $b < c$, $\mathcal{A}$ = any attachment set of $G$

\item[Maximum Spanning Tree:] $c = n$, $b = c - 1$, $\mathcal{A}$ = all acyclic subgraphs of $G$

\item[Maximum Traveling Salesperson:] $c = n$, $b = 2$, $\mathcal{A}$ = all subgraphs of $G$ without non-Hamiltonian cycles

\item[Max Weight Planar Subgraph:] $c = n$, $b = c - 1$, $\mathcal{A}$ = all planar subgraphs of $G$

\item[Max Weight $b$-Matching:] $c = n$, $b < c$, $\mathcal{A}$ = all subgraphs of $G$

\item[Max Weight Matching:] $c = n$, $b = 1$, $\mathcal{A}$ = all subgraphs of $G$
\end{description}

\section{Algorithmic Framework} \label{framework}
In this section we define the Ordinal Greedy algorithm and reveal some of its salient properties. Rather than limit ourselves only to ABC Systems, in this section we consider general graphic {\em independence systems}. An independence system for our setting is a pair $(\mathcal{E, L})$ where $\mathcal{E}$ corresponds to the set of edges in some graph and $\mathcal{L}$ is a collection of subsets of $\mathcal{E}$ such that if $F \in \mathcal{L}$ and $F' \subseteq F$ then $F' \in \mathcal{L}$. The sets in $\mathcal{L}$ are called independent. It is easy to see that all ABC Systems are independence systems because their three constraints possess this heredity property. Let $\mathcal{B}$ denote the set of all subgraphs in which all nodes have degree at most $b$ and let $\mathcal{C}$ denote the set of all subgraphs in which all connected components have size at most $c$. Our Max ABC problem can be restated as: Given a graph $G = (\mathcal{N,E})$, attachment set $\mathcal{A}$, degree limit $b$ and component size limit $c$, compute the maximum weight subgraph in $\mathcal{L = A \cap B \cap C}$.

%_______________________________%
\subsection{The Ordinal Greedy Algorithm} \label{sub:algorithm}
In an ordinal setting, algorithms only have access to a partial ordering, or set of preference orderings, which provide strictly local information about the preferences of each agent. This precludes the use of algorithms which require comparisons between the weights of non-adjacent edges. In fact, it is not difficult to see that no ordinal algorithm can be guaranteed to compute the optimal solution for even simple settings, e.g., forming a matching \cite{anshelevich2016blind}. However, the Ordinal Greedy algorithm defined below performs well in this setting because it relies on strictly local information. Ordinal Greedy starts from the empty set and builds up a sequence of intermediate solutions by adding locally optimal edges at each iteration which do not violate a set of constraints, i.e., preserve independence. To understand how this heuristic is applied to the ordinal setting, we must formalize what it means for an edge to be locally optimal.

\begin{definition} \label{def:undominated} Undominated Edge\\
Given a set $E$ of edges, $(u,v) \in E$ is undominated if for all $(u, x)$ and $(v, y)$ in $E$, $w(u,v) \geq w(u,x)$ and $w(u,v) \geq w(v, y)$.
\end{definition}

At this point it is important to make several observations. First, every edge set $E$ has at least one undominated edge, because its maximum weight edge must be undominated. However, there may be undominated edges which are not globally maximum. Second, for any edge set $E$ it is straightforward to find at least one undominated edge using only the partial ordering $\sigma$ (see \cite{anshelevich2016blind} for details). Undominated edges are either of the form $(u,v)$ where $u$ and $v$ are each other's most preferred neighbor, or form cycles in which each subsequent node is the first choice of the previous one, and thus all edges in the cycle have the same weight.

What follows is a general purpose Ordinal Greedy algorithm, which starts from the empty set and iteratively selects undominated edges from the set of remaining edges which do not violate the constraints in question. The algorithm uses the partial ordering $\sigma$ to determine which edges are undominated at each iteration. The algorithm concludes when there are no edges left which can be added to the subgraph without violating the constraints, so the final solution $S$ is maximal in this sense.

\begin{algorithm}[htb]
\caption{Ordinal Greedy}
\label{alg:ordinal_greedy}
  \textbf{Input:} Edge set $\mathcal{E}$, partial ordering $\sigma$, collection of valid subgraphs $\mathcal{L}$\\
  Initialize $S =  \emptyset $, $E = \mathcal{E}$ \;
   \While {$E \ne \emptyset$} {
     Pick an undominated edge $e = (u,v) \in E$ and add it to the intermediate solution: $S \leftarrow S + e$ \;
     Remove $e$ from $E$ \;
     Remove all edges $f$ from $E$ such that $S + f \notin \mathcal{L}$ \;
    }
  \textbf{Output:} Return $S$
\end{algorithm}

We refer to the iteration at which an edge $e = (u,v)$ is removed from $E$ as the \textbf{critical iteration} of $e$. When the inputs to our algorithm $(\mathcal{E,L})$ characterize an ABC System, there are exactly four cases which may occur at the critical iteration of edge $e$:\\
1) $e$ is added to the ordinal greedy solution $S$\\
2) $e$ is removed from $E$ because $S + e \notin \mathcal{A}$\\
3) $e$ is removed from $E$ because $S + e \notin \mathcal{B}$ (where $\mathcal{B}=$ sets of edges with any degree $\leq b$)\\
4) $e$ is removed from $E$ because $S + e \notin \mathcal{C}$ (where $\mathcal{C}=$ sets of edges with any component size $\leq c$)

For cases 2-4 we say $e$ was \textbf{eliminated due to $\mathcal{A, B,}$ or $\mathcal{C}$}. If an edge $e = (u,v)$ was eliminated due to $\mathcal{A}$, the attachment property implies there must be a $(u,v)$-path in the intermediate solution at its critical iteration. In other words, $u$ and $v$ are already in the same connected component in $S$ at this iteration. If $e = (u,v)$ was eliminated due to $\mathcal{B}$, either $u$ or $v$ must already have degree exactly $b$ at this iteration. If $e = (u,v)$ was eliminated due to $\mathcal{C}$, $u$ and $v$ must already be in disjoint connected components whose cumulative size is greater than $c$. Note that in these three cases, an edge can only be eliminated if at least one adjacent edge of equal or greater weight has already been added to the intermediate solution, and all adjacent edges already added to the intermediate solution must be of equal or greater weight (since only undominated edges are added to our solution).

There are limiting values of $\mathcal{A}$, $b$, and $c$ for which elimination due to these constraints cannot occur. When $c = n$, no edge can be eliminated due to $\mathcal{C}$ because all nodes can be in the same connected component. Additionally, when $\mathcal{A}$ is the set of all subgraphs of $G$, no edge can be eliminated due to $\mathcal{A}$. If adding an edge would violate more than one constraint, we say that it was eliminated in order of priority $\mathcal{C} \rightarrow \mathcal{B} \rightarrow \mathcal{A}$. For example, when $b = c - 1$, no edge can be eliminated due to $\mathcal{B}$, because for a node to reach degree $b$ the size of its component must be exactly $c$ and we say that any incident edge would be eliminated due to $\mathcal{C}$. As the following sections show, the approximation factor of Ordinal Greedy for an ABC System depends on which cases of elimination can occur.

Notice that the performance of the Ordinal Greedy algorithm can deviate significantly from the Omniscient Greedy algorithm in the full-information setting (which we call ``Omniscient Greedy" because it knows the underlying edge weights and can choose the edge with maximum weight at each iteration). Consider the following example.

\begin{example} \label{example:MST}
\em{
Suppose the graph $G = (\mathcal{N,E})$ is constructed as follows.
Let $\mathcal{N} = \{u_1, ... u_k, v_1, ...v_k\}$.
Let $w(u_i,v_i) = 1+\epsilon$ for $i \leq k$ for some infinitesimal $\epsilon$.
Let $w(u_i,u_j) = 1$ for all $i \neq j$. Let $w(v_i,v_j) = \epsilon$ for all $i \neq j$.
Let all other edges have weight 0. Consider the ABC System corresponding to finding a Maximum-Weight Spanning Tree. It is clear that Omniscient Greedy will find the optimum solution with weight $w(OPT_w) = k(1+\epsilon)+(k-1)$.

Now consider Ordinal Greedy. Suppose Ordinal Greedy begins by selecting $(u_i,v_i)$ for $i \leq k$, which are all undominated at the beginning of the algorithm.
Once these edges have been selected, edges of the form $(u_i,u_j)$ and $(v_i,v_j)$ become undominated for $i \neq j$. Now, if an edge $(u_i,u_j)$ or $(v_i,v_j)$ is selected, the other must be eliminated at that iteration, since taking it would form a cycle. Notice that since we only have access to ordinal information, there is {\em no possible way} for Ordinal Greedy to tell which of these edges is better: they are both edges which are most preferred by their endpoints, even though one secretly has weight $1$ and the other only $\epsilon$. In other words, these edges are incomparable in the partial preference order $\sigma$. Suppose Ordinal Greedy proceeds by selecting $(v_i,v_{i+1})$ for $i < k$. Then the Ordinal Greedy solution formed has weight $w(S) = k(1+\epsilon) + (k-1)\epsilon$. This example shows that (in the limit) it is not possible for Ordinal Greedy to always result in solutions better than a factor of 2 away from optimum, even though Omniscient Greedy can easily compute the true optimum solution. As we show in this paper, however, despite its knowledge handicap, Ordinal Greedy can often produce surprisingly good results.
}
\end{example}

\subsection{Properties of Ordinal Greedy} \label{sub:binary}
For any independence system in the full-information setting, the Omniscient Greedy algorithm has been shown to achieve its worst approximation on an instance with a \textbf{binary weight function} $\bar{w} : \mathcal{E} \rightarrow \{0,1\}$ \cite{korte1978analysis}. However, previous proofs have relied crucially on the fact that Omniscient Greedy selects edges in strictly non-increasing order by weight, which is not possible with only ordinal information. We offer a new proof to show that even in the ordinal setting, Ordinal Greedy always achieves its worst approximation factor on an instance with a binary weight function $\bar{w} : \mathcal{E} \rightarrow \{0,1\}$ for any graphic independence system. This theorem will allow us to prove approximation bounds for ABC and AB Systems later in this paper.

\begin{theorem} \label{thm:binary}
For any graphic independence system $(\mathcal{E,L})$, for any instance $(w, \sigma)$ with weight function $w: \mathcal{E} \rightarrow \mathbb{R}^+$ and partial ordering $\sigma$ consistent with $w$, there exists an instance $(\bar{w}, \sigma)$ with weight function $\bar{w}: \mathcal{E} \rightarrow \{0,1\}$ such that $\sigma$ is consistent with $\bar{w}$ and the worst-case ratio of the optimal solution to an Ordinal Greedy solution is at least as large as for $(w, \sigma)$.
\end{theorem}

\begin{proofsketch} Before we begin the proof, we provide a short proof sketch.
Suppose on instance $(w, \sigma)$ the ratio between the optimal solution $OPT_w$ and solution $S$ constructed by Ordinal Greedy is $\frac{w(OPT_w)}{w(S)} = \delta$. Our goal is to construct a binary weight function $\bar{w}$ such that $\frac{\bar{w}(OPT_{\bar{w}})}{\bar{w}(S)} \geq \delta$. When $\delta$ is infinite, constructing $\bar{w}$ is straightforward, so we only consider finite values of $\delta$. First we create a weight function $\hat{w}$ by raising the weights of all edges not in $S$ as much as possible without altering the weights of the edges of $S$, such that $\sigma$ remains consistent with $\hat{w}$. Since Ordinal Greedy selects $S$ and none of its edge weights have changed, and the edge weights of $OPT_w$ cannot have decreased, then $\frac{\hat{w}(OPT_w)}{\hat{w}(S)} \geq \delta$. From $\hat{w}$ we carefully create $\bar{w}$ by proving that there must exist a subset of edges to which we can assign weight 1 and let all other edges have weight 0, such that $\sigma$ is consistent with $\bar{w}$ and $\frac{ \bar{w}( OPT_{\bar{w}} ) }{\bar{w}(S)} \geq \frac{\bar{w}(OPT_w)}{\bar{w}(S)} \geq \delta$.
\end{proofsketch}

\begin{proof}
Recall that a partial ordering $\sigma$ is {\em consistent} with weight function $w$ if for all $x,y,z \in \mathcal{N}$, if $x$ prefers $y$ to $z$ in $\sigma$ then $w(x,y) \geq w(x,z)$.
We will now show that for any instance $(w, \sigma)$ where $w: \mathcal{E} \rightarrow \mathbb{R}^+$ for which Ordinal Greedy provides a $\delta$-approximation for $\delta > 0$ in the worst case, there exists an instance $(\bar{w}, \sigma)$ where $\bar{w}: \mathcal{E} \rightarrow \{0,1\}$ for which Ordinal Greedy provides no better than a $\delta$-approximation in the worst case. Given any weight function $w$ we now construct a binary weight function $\bar{w}$ such that the approximation factor is at least as large and $\sigma$ is still consistent with $\bar{w}$.

\begin{observation} \label{ob:consistent1}
\emph{For an independence system $(\mathcal{E,L})$ the solution $S$ computed by Ordinal Greedy depends only on $\sigma$, not the edge weights. Therefore, if $\sigma$ is consistent with $w$ and $\hat{w}$, the possible solutions $S$ are the same for instances $(w, \sigma)$ and $(\hat{w},\sigma)$. However, $\hat{w}(S)$ and $w(S)$ may differ.}
\end{observation}

Therefore, given a worst possible solution $S$ constructed by Ordinal Greedy for $(w,\sigma)$, our goal is to take the weight function $w$ and construct a binary weight function $\bar{w}$ such that $\frac{w(OPT_w)}{w(S)} \leq \frac{\bar{w}(OPT_{\bar{w}})}{\bar{w}(S)}$, and $\sigma$ is still consistent with $\bar{w}$ (and thus $S$ can still be produced by Ordinal Greedy for the instance with weights $\bar{w}$). Recall that $OPT_w$ is the optimum (maximum-weight) solution for weights $w$. Since by definition $\bar{w}(OPT_{\bar{w}}) \geq \bar{w}(OPT_w)$, it is enough to show that $\delta = \frac{w(OPT_w)}{w(S)} \leq \frac{\bar{w}(OPT_w)}{\bar{w}(S)}$.

Suppose there is some edge $e \notin S$, such that no edge in $S$ is known to be greater than or equal to it in the partial ordering $\sigma$. If such an edge exists, we can let $\bar{w}$ be the weight function such that $e$ and all edges known to be greater than or equal to $e$ in $\sigma$ have weight 1, and all other edges have weight 0. Clearly, $\sigma$ remains consistent with $\bar{w}$. To see this, consider any $x,y,z \in \mathcal{N}$ such that $x$ prefers $y$ to $z$. Then either both $w(x,y)$ and $w(x,z)$ dominate $e$ in $\sigma$, so $\bar{w}(x,y)=\bar{w}(x,z)=1$, or neither do, in which case $\bar{w}(x,y)=\bar{w}(x,z)=0$. The only other case is that $w(x,y)$ dominates $e$ and $w(x,z)$ does not (since we know that $w(x,y)$ dominates everything that $w(x,z)$ does), and then $1= \bar{w}(x,y)>\bar{w}(x,z)=0$. In all cases, $\bar{w}(x,y)\geq\bar{w}(x,z)$, so $\sigma$ is consistent with $\bar{w}$. This means $S$ remains the same set of edges, and since $\bar{w}(S)=0$, the approximation factor becomes unbounded. This means for any instance where there is some edge $e \notin S$ for which no edge in $S$ dominates it in $\sigma$, we can always create a weight function $\bar{w}$ with an approximation factor at least as large. Therefore, for the rest of the proof we assume that for every edge $e \notin S$ there exists some edge in $S$ known to be at least as great by the partial ordering $\sigma$.

For our greedy solution $S$, fix an ordering $\{s_1, s_2, ..., s_m\}$ over the edges of $S$ in non-increasing order by weight so that $w(s_1) \geq w(s_2) \geq ... \geq w(s_m)$. Construct weight function $\hat{w}$ by increasing the weight of each edge not in $S$ to be equal to the weight of the smallest-weight edge $s_i \in S$ known to be greater than or equal to it in the partial ordering $\sigma$. Note that by our assumption above, such an edge $s_i$ always exists.

\begin{claim}\label{claim:consistent}
$\sigma$ is consistent with $\hat{w}$.
\end{claim}

\begin{proof}
Consider any two adjacent edges $(x,y)$ and $(x,z)$ where $x$ prefers $y$ to $z$. Let $(u,v)\in S$ be a smallest edge known by $\sigma$ to have weight at least $w(x,y)$. Then $(u,v)$ is also known to have weight at least $w(x,z)$ since $x$ prefers $y$ to $z$. Therefore, the smallest edge of $S$ known to have weight at least $w(x,z)$ is either $(u,v)$ or has weight smaller than $w(u,v)$. After increasing the weights, $\hat{w}(x,y) = w(u,v) \geq \hat{w}(x,z)$. Therefore if $x$ prefers $y$ to $z$ in $\sigma$, then $\hat{w}(x,y) \geq \hat{w}(x,z)$. By definition, $\sigma$ remains consistent with $\hat{w}$.
\end{proof}

The above process of forming $\hat{w}$ forms an assignment of edges: consider every edge $e$ to be assigned to the smallest edge in $s_i \in S$ known to be larger than or equal to it in the partial ordering, where $\hat{w}(e) = \hat{w}(s_i)$. Now alter these assignments so that if $\hat{w}(s_i) = \hat{w}(s_{i+1})$ then all edges with this weight, including $s_i$, are assigned to $s_{i+1}$. Let $r_i$ be the number of edges of $OPT_w$ assigned to $s_i$. Note that if $\hat{w}(s_i) = \hat{w}(s_{i+1})$ then $r_i = 0$.

\begin{lemma} \label{lemma:assignedweights}
$\sum \limits_{i = 1}^m r_i \cdot \hat{w}(s_i) \geq \delta \sum \limits_{i = 1}^m \hat{w}(s_i)$
\end{lemma}

\begin{proof}
On the left side of the inequality, the product $r_i \cdot \hat{w}(s_i)$ denotes the total weight of the edges of $OPT_w$ assigned to $s_i$ after having their weight increased. This sum over $i \leq m$ computes $\hat{w}(OPT_w)$, the total weight of the optimal solution over $w$, evaluated by $\hat{w}$. On the right hand side, the summation yields the total weight of the edges in the greedy solution, $\hat{w}(S)$ multiplied by $\delta$. Since $\sigma$ is consistent with $\hat{w}$ by Claim \ref{claim:consistent} the greedy solution $S$ remains the same, and since none of the weights of edges of $S$ were altered $\hat{w}(S) = w(S)$. By construction, $\hat{w}(OPT_w) \geq w(OPT_w)$ because the weights of edges of $OPT_w$ could only have been increased. And so $\hat{w}(OPT_w) \geq \frac{w(OPT_w)}{w(S)} \hat{w}(S) = \delta \cdot \hat{w}(S)$.
\end{proof}

We now demonstrate that we can alter the weights of $\hat{w}$ to create a binary weight function $\bar{w}: \mathcal{E} \rightarrow \{0,1\}$ such that $\bar{w}(OPT_w) \geq \delta \cdot \bar{w}(S)$. All changes to the weights keep $\sigma$ consistent with $\bar{w}$, so that $S$ remains a solution of Ordinal Greedy.

\begin{lemma}
There exists some $k \leq m$ such that $\sum \limits_{i=1}^k r_i \geq \delta k$.
\end{lemma}

\begin{proof}
Suppose to the contrary that $\sum \limits_{i=1}^k r_i < \delta k$  for all $k \leq m$. We show by induction that this implies $\sum \limits_{i = 1}^m r_i \hat{w}(s_i) < \delta \sum \limits_{i = 1}^m \hat{w}(s_i)$, which yields a contradiction to Lemma \ref{lemma:assignedweights}. Specifically, we will show that for every $j\leq m$,

\begin{equation}\label{ineq.1}
\sum \limits_{i=1}^j r_i \hat{w}(s_i) \leq \delta\sum_{i=1}^{j-1} \hat{w}(s_i) + [\sum \limits_{i=1}^j r_i - \delta(j-1)] \hat{w}(s_j).
\end{equation}

When applied to $j=m$, this gives us the result that $\sum \limits_{i = 1}^m r_i \hat{w}(s_i) \leq \delta\sum_{i=1}^{m-1} \hat{w}(s_i) + [\sum \limits_{i=1}^m r_i - \delta(m-1)] \hat{w}(s_m)$. Since $\sum \limits_{i=1}^m r_i < \delta m$, then $\sum \limits_{i=1}^m r_i - \delta(m-1) < \delta$, and thus the right hand side of the above inequality is strictly less than $\delta\sum_{i=1}^{m} \hat{w}(s_i)$, which gives us a contradiction with Lemma \ref{lemma:assignedweights}, as desired. Note that here we use the fact that $\hat{w}(s_m)>0$ without loss of generality; if this were not the case then we can make the same argument for $j$ being the largest integer such that $\hat{w}(s_j)>0$. Thus all that is left is to prove Inequality (\ref{ineq.1}).

We proceed by induction. The base case for $j=1$ is trivially true. Now assume that Inequality (\ref{ineq.1}) holds for $j$, and we will prove it for $j+1$. Then,

$$\sum \limits_{i=1}^{j+1} r_i \hat{w}(s_i) \leq \delta\sum_{i=1}^{j-1} \hat{w}(s_i) + [\sum \limits_{i=1}^j r_i - \delta(j-1)] \hat{w}(s_j) + r_{j+1}\hat{w}(s_{j+1}). $$

Let $\xi = \delta j - \sum \limits_{i=1}^j r_i$. Since by our assumption $\sum \limits_{i=1}^j r_i < \delta j$, we know that $\xi > 0$. Suppose in the right-hand side of the above inequality, we increase the coefficient of $\hat{w}(s_j)$ by $\xi$, and decrease the coefficient of $\hat{w}(s_{j+1})$ by $\xi$. Since $\hat{w}(s_j)\geq \hat{w}(s_{j+1})$, this only makes the quantity larger. Thus, we obtain that:

$$\sum \limits_{i=1}^{j+1} r_i \hat{w}(s_i) \leq \delta\sum_{i=1}^{j} \hat{w}(s_i) + [\sum \limits_{i=1}^{j+1} r_i - \delta j]\hat{w}(s_{j+1}), $$
as desired. This proves Inequality (\ref{ineq.1}) for every $j\leq m$, and thus completes the proof of this Lemma.
\end{proof}

We now use this value of $k$ to construct the binary weight function $\bar{w}$. Take the smallest $k$ such that $\sum \limits_{i=1}^k r_i \geq \delta k$ and let $\bar{w}$ be the weight function where for all $i \leq k$, $\bar{w}(s_i) = 1$, $\bar{w}(e) = 1$ for all $e$ assigned to $s_i$, and all other edges have weight 0. Note that $r_k>0$ since if $r_k = 0$ and $\sum \limits_{i=1}^k r_i \geq \delta k$, then $\sum \limits_{i=1}^{k-1} r_i \geq \delta (k-1)$.

We now argue that $\sigma$ is still consistent with $\bar{w}$. Consider any adjacent edges $(x,y),(x,z)$ such that $x$ prefers $y$ to $z$, and suppose that $(x,y)$ is assigned to edge $s_i$, while $(x,z)$ is assigned to edge $s_j$. First consider the case when $w(s_i)\neq w(s_j)$. It must be that $i < j$, since the set of edges dominating $(x,y)$ is a subset of edges dominating $(x,z)$, and edges are assigned to a smallest edge of $S$ dominating them. Then $\bar{w}(x,y) \geq \bar{w}(x,z)$, as desired. If instead $w(s_i)= w(s_j)$, then by construction of our assignment we have that $i=j$, so $\bar{w}(x,y) = \bar{w}(x,z)$. Therefore $\sigma$ is consistent with $\bar{w}$ because for any adjacent edges $(x,y),(x,z)$ such that $x$ prefers $y$ to $z$ it must be that $\bar{w}(x,y) \geq \bar{w}(x,z)$.

Thus we now have a binary weight function $\bar{w}$ such that $S$ is a possible solution of Ordinal Greedy, since $\sigma$ is consistent with $\bar{w}$. By definition of $\bar{w}$, we know that $\bar{w}(OPT_w) = \sum_{i=1}^k r_i$, and $\bar{w}(S) = k$. Due to our choice of $k$, we thus have that $\bar{w}(OPT_w) \geq \delta \cdot \bar{w}(S)$.

This concludes our proof that for any instance $(w, \sigma)$ where $w: \mathcal{E} \rightarrow \mathbb{R}^+$ for which Ordinal Greedy provides a $\delta$-approximation for $\delta > 0$ in the worst case, there exists an instance $(\bar{w}, \sigma)$ where $\bar{w}: \mathcal{E} \rightarrow \{0,1\}$ for which Ordinal Greedy provides no better than a $\delta$-approximation in the worst case.
\end{proof}

Another nice property is that every solution constructed by Ordinal Greedy is pairwise stable. Pairwise stability means that no pair of agents $x,y$ has incentive to collude to add edge $(x,y)$ by each giving up some of their edges in the Ordinal Greedy solution $S$. Either this exchange would decrease the total utility of one of the agents, or adding $(x,y)$ is infeasible even after sacrificing the other edges. Here we assume that the utility of a node $x$ in solution $S$ is simply the total weight of edges in $S$ incident on $x$.

\begin{theorem} \label{thm:stability}
Any solution $S$ constructed by Ordinal Greedy on an independence system is pairwise stable.
\end{theorem}

\begin{proof}
Let $(x,y) \notin S$ and let $F_x$ and $F_y$ be any set of edges in $S$ adjacent to $x$ and $y$ excluding $(x,y)$. If $x$ and $y$ can improve their individual utilities by adding $(x,y)$ at the expense of removing all of the edges in $F_x \cup F_y$, this means $w(x,y) > w(F_x)$ and $w(x,y) > w(F_y)$. However, this clearly implies $w(x,y)$ is larger than the weight of each individual edge in $F_x$ and $F_y$. If $(x,y) \notin S$, then its critical iteration must have occurred before any of the adjacent edges in $F_x$ and $F_y$ were added to the Ordinal Greedy solution. Therefore $S - F_x - F_y + (x,y)$ cannot be a feasible solution.
\end{proof}
%_______________________________%

%_______________________________%
\section{Ordinal Approximation for ABC Systems} \label{abc}
In this section we bound the worst-case performance of Ordinal Greedy compared to the optimal solution for any ABC System. We use $\alpha_{_{ABC}}$ to denote the approximation factor, or the ratio of the optimal solution to the worst possible Ordinal Greedy solution for any ABC System.

Unlike Example \ref{example:MST} in Section \ref{sub:algorithm} for the maximum spanning tree problem, Ordinal Greedy does not provide a constant approximation factor for all ABC Systems. However, it does always provide a finite approximation which depends on the degree limit $b$. To simplify notation, since the optimal solution here is only evaluated using the same weight function used to generate it, we refer to the total weight of the optimal solution $w(OPT_w)$ as $w(OPT)$. Here we show that $\alpha_{_{ABC}} \leq b+1$ for any ABC System and provide a family of examples where $\frac{w(OPT)}{w(S)} = b+1$ to show that $b+1$ is a tight bound on the approximation factor. In later sections, we explore classes of ABC Systems in which Ordinal Greedy achieves a better worst-case approximation.

Note that this result is quite general. As we discussed, ABC Systems include many varied constraints, some quite difficult to approximate. Our result in this section states that, even for extremely complex $\mathcal{A}$ and constraints on component size $c$, as long as the maximum allowed degree of any node is small, then it is possible to form a good approximation to the true optimum solution while only knowing ordinal information instead of the true edge weights.

\begin{theorem} \label{thm:ABC}
For any ABC System, the Ordinal Greedy algorithm always produces a solution within a factor of (b+1) of the optimal solution, and this bound is tight.
\end{theorem}

\begin{proofsketch} Before we begin the detailed proof, we provide a short proof sketch.
We proceed via a charging argument. We wish to charge the weight of the edges of $OPT$ to the edges of $S$ such that all edges of $OPT$ are fully charged somewhere, and no edge of $S$ receives a charge greater than $b+1$ times its weight. However, unlike Omniscient Greedy in the full-information setting, we cannot assume that any eliminated edge of $OPT$ has weight smaller than all edges of the Ordinal Greedy solution $S$ which were added before its critical iteration. This prohibits us from using the methods in previous work. Thankfully, due to Theorem \ref{thm:binary} we know that if Ordinal Greedy produces a solution within a factor of $(b+1)$ of optimal for all instances with binary weight functions, then this holds for all instances. We therefore assume that all weights are $\{0,1\}$, and can now charge any edge of $OPT$ to any weight 1 edge of $S$, but must ensure that no weight is charged to any edges of $S$ with weight 0.

To ensure no edge of $S$ is charged more than $b+1$ times its weight, we look at the connected components of $S$ with only weight 1 edges, and charge all edges of $OPT$ to these components. Specifically, we design the following charging scheme. Let $(u,v)$ be an edge of $OPT$ where $w(u,v) = 1$. Let $P_u$ and $P_v$ be the connected components containing $u$ and $v$ in the subgraph of $S$ containing only weight 1 edges. We charge the weight of $w(u,v)$ between $P_u$ and $P_v$ based on what occurs at the critical iteration of $(u,v)$. If $(u,v) \in S$ then $P_u = P_v$, so charge its full weight to this component. If $(u,v)$ was eliminated due to $\mathcal{A}$, charge its full weight to either $P_u$ or $P_v$ arbitrarily. Note that while the attachment property of $\mathcal{A}$ ensures that $u$ and $v$ are in the same connected component in $S$ at this iteration, this does \emph{not} imply that $P_u=P_v$ because all $(u,v)$-paths in $S$ may contain a weight 0 edge. If $(u,v)$ was eliminated due to $\mathcal{B}$, one of its endpoints must have a degree of exactly $b$ in $S$ at this iteration, so charge its full weight to the component containing this endpoint. If $(u,v)$ was eliminated due to $\mathcal{C}$, we split the charge between $P_u$ and $P_v$ based on the size of the connected components in $S$ containing $u$ and $v$ at this iteration. Let $q_u$ and $q_v$ be the sizes of the connected components in $S$ containing $u$ and $v$ at the critical iteration of $(u,v)$. Charge $\frac{q_u - 1}{q_u + q_v - 2}$ to $P_u$ and $\frac{q_v - 1}{q_u + q_v - 2}$ to $P_v$. In all four cases we have ensured the full weight of the edge of $OPT$ has been charged between the components containing its endpoints.

The rest of the proof involves arguing that each such component $P_u$ of size $p$ is charged a total of at most $(b+1)(p-1)$ using the above charging scheme. Since such a component must contain at least $p-1$ edges with weight 1, this completes the proof of the upper bound because it shows that the total number of edges in $OPT$ with weight 1 is at most $b+1$ times the number of such edges in $S$.

We then provide a family of examples to show this bound is tight. Omniscient Greedy has the same worst-case solution as Ordinal Greedy on this family of examples, so our ordinal approximation competes well despite its knowledge handicap. Note that in the example yielding the lower bound of $\alpha_{_{ABC}} \geq b+1$, all edges are eliminated due to $\mathcal{C}$. The following section demonstrates that when the component size constraint is relaxed by allowing $c = n$, the approximation factor improves significantly.
\end{proofsketch}

\begin{proof}
First, we prove the upper bound $\alpha_{ABC} \leq b+1$. Since all ABC Systems are independence systems, by Theorem \ref{thm:binary} it is sufficient to show that $\alpha \leq b+1$ for all instances with integral weight functions $w : \mathcal{E} \rightarrow \{0,1\}$ to provide an upper bound on $\alpha$ for all instances.

We proceed via a charging argument. Let $w$ be a binary weight function $w: \mathcal{E} \rightarrow \{0,1\}$, so all edges in our graph are either weight 1 or weight 0. Given the Ordinal Greedy solution $S$ constructed for some instance $(w,\sigma)$ we consider the subgraph $\bar{S} \subseteq S$ with only weight 1 edges and partition $\bar{S}$ into connected components. We then construct a charging scheme which charges the full weight of all the edges in $OPT$ to these components. The total charge over all these components represents the total weight of $OPT$, while the total weight of these components is equal to the weight of the greedy solution. If our charging scheme guarantees that the full weight of all edges of $OPT$ has been charged to the components, and no component of $\bar{w}$ is charged more than $(b+1)$ times the total weight of its edges, then we can sum over these components to show that $w(OPT) \leq (b+1)w(\bar{S}) = (b+1)w(S)$.

Let $\mathcal{P}$ be the set of connected components in $\bar{S}$. Let $P \in \mathcal{P}$ be any connected component in this subgraph of size $p \geq 1$ and total weight $w(P)$. We make two obvious, but critical observations about each component $P$.

\begin{observation} \label{ob:ABC_componentweight}
\emph{For any component $P$, its total weight $w(P)$ is at least $p-1$, the weight of its maximum spanning tree.}
\end{observation}

\begin{observation} \label{ob:ABC_path}
\emph{There is a path between any two nodes in $S$ with only weight 1 edges if and only if the nodes are in the same component $P$.}
\end{observation}

\begin{claim} \label{claim:ABC_charging}
\emph{There exists a charging scheme which charges the full weight of all edges of $OPT$ to the components of $\mathcal{P}$ such that the total charge to any component $P$ of size $p$ is at most $(b+1)(p-1)$.}
\end{claim}

\begin{proof}

We need only charge the edges of $OPT$ such that $w(u,v) = 1$ because the weight 0 edges do not contribute to the value of $w(OPT)$.

%Recall that at its critical iteration it could be included or eliminated due to A, B, C...\\

\noindent\textbf{Charging Scheme:}
Let $(u,v)$ be an edge of $OPT$ where $w(u,v) = 1$. Let $P_u$ and $P_v$ be the connected components containing $u$ and $v$ in the subgraph $\bar{S} \subseteq S$ containing only weight 1 edges. We charge the weight of $w(u,v)$ between $P_u$ and $P_v$ based on what occurs at the critical iteration of $(u,v)$. If $(u,v) \in S$ then $P_u = P_v$, so charge its full weight to this component. If $(u,v)$ was eliminated due to $\mathcal{A}$, charge its full weight to either $P_u$ or $P_v$ arbitrarily. Note that while the attachment property of $\mathcal{A}$ ensures that $u$ and $v$ are in the same connected component in $S$ at this iteration, this does \emph{not} imply that they are in the same connected component in $\mathcal{P}$ because all $(u,v)$-paths in $S$ may contain a weight 0 edge. If $(u,v)$ was eliminated due to $\mathcal{B}$, one of its endpoints must have a degree of exactly $b$ in $S$ at this iteration, so charge its full weight to the component containing this endpoint. If $(u,v)$ was eliminated due to $\mathcal{C}$, we split the charge between $P_u$ and $P_v$ based on the size of the connected components in $S$ containing $u$ and $v$ at this iteration. Let $q_u$ and $q_v$ be the sizes of the connected components in $S$ containing $u$ and $v$ at the critical iteration of $(u,v)$. Charge $\frac{q_u - 1}{q_u + q_v - 2}$ to $P_u$ and $\frac{q_v - 1}{q_u + q_v - 2}$ to $P_v$. In all four cases we have ensured the full weight of the edge of $OPT$ has been charged between the components containing its endpoints.

The weight of each edge $(u,v) \in OPT$ has been charged exclusively to the components containing its endpoints, $P_u$ and $P_v$. Therefore, to determine the maximum possible charge to any component $P$, we bound the charge from edges of $OPT$ with one or both endpoints in $P$ and show this is at most $(b+1)(p-1)$.

\textbf{Case 1)} $p > b$\\
All nodes have at most $b$ adjacent edges in $OPT$, so the maximum charge to any component $P$ with $p$ nodes is $b \cdot p$. If $p > b$ then the total charge on $P$ is at most $b \cdot p \leq b \cdot p + (p - (b+1)) = (b+1)(p-1)$.

\textbf{Case 2)} $p \leq b$\\
Let $u_0$ be the first node in $P$ to have an edge of weight 0 added adjacent to it in $S$ at some iteration of Ordinal Greedy. By definition, at any time before the critical iteration of this weight 0 edge $u_0$ cannot be in the same component in $S$ as any edge of weight 0. If there is no node in $P$ with an adjacent weight 0 edge in $S$, let $u_0$ be any arbitrary node in $P$. Let a \textbf{Type 2} edge be an edge of $OPT$ with weight 1 which is incident to $u_0$, but not to any other node in $P$. Let a \textbf{Type 1} edge be all edges of $OPT$ of weight 1 which are not Type 1, including all edges with both endpoints in $P$ and those with a single endpoint in $P$ which is not $u_0$.

Since all nodes have degree at most $b$ in $OPT$ and there are $(p-1)$ nodes other than $u_0$, it is clear that Type 1 edges cumulatively contribute a charge of at most $b \cdot (p-1)$ to $P$. Here we show that Type 2 edges contribute a total charge of at most $p - 1$, limiting the total charge to any component $P$ to at most $b \cdot (p - 1) + (p - 1) = (b+1)(p - 1)$.

Let $(u_0, v)$ be an edge of $OPT$ where $w(u_0,v) = 1$. The critical iteration of $(u_0,v)$ must be before the critical iteration of the first weight 0 edge incident to $u_0$ in $S$. This is because at the iteration the weight 0 edge was added to $S$ it had to be undominated, so $(u_0,v)$ could not still have been adjacent to it in the set of available edges $E$.

Let $P_v$ denote the component containing $v$. We look at the four cases of charging based on the critical iteration of $(u_0, v)$ to show that Type 2 edges contribute a charge of at most $(p-1)$ to $P$.

If $(u_0, v) \in S$, then clearly this is Type 1 because $P = P_v$ and so it has already been charged to $P$. Likewise, if $(u_0, v)$ was eliminated due to $\mathcal{A}$ then its endpoints must be in the same component $P = P_v$. This is because the attachment property ensures that $u$ and $v$ are in the same connected component in $S$ at this iteration, and by the definition of $u_0$ it cannot yet have a path to any edge of weight 0. Therefore any path from $u_0$ to $v$ in $S$ must contain only weight 1 edges, meaning $P = P_v$. If $(u_0, v)$ was eliminated due to $\mathcal{B}$ then its full weight is either charged to $P$ or to $P_v$. For its weight to be charged to $P$, $u_0$ must have degree $b$ at this iteration. For $(u_0, v)$ to have weight 1 this means all $b$ of the edges incident to $u_0$ in $S$ must have weight 1 because otherwise they could not have been undominated before the critical iteration of $(u_0, v)$. However, this would mean that all neighbors of $u_0$ are in $P$, so $p > b$. Therefore, if $p \leq b$ there can be no Type 2 edges charged to $P$ which were eliminated due to $\mathcal{B}$.

We can now see that the only Type 2 edges charged to $P$ are those which are eliminated due to $\mathcal{C}$. For each of these edges $P$ is charged $\frac{q_{u_0} - 1}{q_{u_0} + q_v - 2}$. For $(u_0,v)$ to be eliminated due to $\mathcal{C}$ this means the combined sizes of the disjoint components in $S$ containing $u_0$ and $v$ must be at least $c$, otherwise the edge $(u_0,v)$ would still be a valid edge to add. In other words, $q_{u_0} + q_v > c$ so $q_{u_0} + q_v \geq c-1$. And since no weight 0 edge may yet be adjacent to the component in $S$ containing $u_0$ at this iteration, we know that $q_{u_0} \leq p$. Therefore, $\frac{q_{u_0} - 1}{q_{u_0} + q_v - 2} \leq \frac{p - 1}{c - 1}$, so the charge from each Type 2 edge eliminated due to $\mathcal{C}$ is at most $\frac{p - 1}{c - 1}$. And since there can be at most $b$ Type 2 edges, the total charge they contribute to $P$ is at most $b \ \frac{p - 1}{c - 1} = (p-1)\frac{b}{c-1} \leq p-1$ because $b \leq c - 1$ (the maximum degree can never be more than the component size).

This leaves us with a total charge to $P$ of at most $b(p-1) + (p-1) = (b+1)(p-1)$ when $p \leq b$. Together with Case 1, we have shown that the total charge to each component $P \in \mathcal{P}$ of size $p$ at most $(b+1)(p-1)$. This concludes the proof of Claim \ref{claim:ABC_charging}.
\end{proof}

By summing the charge over all components $P \in \mathcal{P}$ we get $w(OPT) \leq (b+1) \sum\limits_{P \in \mathcal{P}}(p-1) \leq (b+1) \sum\limits_{P \in \mathcal{P}} w(P) = (b+1) w(S)$ from the above claim. In other words, $\alpha_{_{ABC}} = \max\limits_{w,\sigma} \frac{w(OPT)}{w(S)} \leq b+1$.

To show that the above result is tight, consider the following ABC System. This system represents the problem of hedonic coalition formation with additive separable symmetric preferences ($b = c-1$ and $\mathcal{A}$ = all subgraphs of $G$), where agents are partitioned into coalitions (cliques) and each agent's total utility is the sum of its utility for being matched with all other agents in its coalition.

\paragraph{Example 2}
Suppose that $n = c^2$, $b = c - 1$, and $\mathcal{A}$ = all subgraphs of G. In other words, the only constraint is that all components must be of size at most $c=\sqrt{n}$.
Label the nodes $u_{ij}$ for $i \in [1,c]$ and $j \in [1,c]$.
Let $w(u_{i1},u_{ij}) = 1$ for all $j > 1$.
Let $w(u_{i1}, u_{k1}) = 1+\epsilon$ for all $k = i + 1$ for some infinitesimal $\epsilon$.
Let all other edges have weight 0.

For appropriate choices of preferences $\sigma$, Ordinal Greedy may select each of the edges $(u_{i1}, u_{k1})$ for all $k = i+1$ before selecting any others, creating a path of length $c-1$. Thus, all of the weight 1 edges are eliminated due to $\mathcal{C}$. The optimal solution is to select each of the edges $(u_{i1},u_{ij}) = 1$ for all $j > 1$. This yields $w(S) = (c - 1)(1+ \epsilon)$ while $w(OPT) = b \cdot c$ because $OPT$ consists of $c$ stars with $b$ edges each. Therefore as $\epsilon \rightarrow 0$, $\frac{w(OPT)}{w(S)} = b \cdot \frac{c}{c-1} = b+1$ because $b = c - 1$. As $\alpha$ is an upper bound on the ratio between the optimal solution and the greedy solution for any instance, we have $\alpha \geq b+1$.
%\end{proof}

This concludes our proof of Theorem \ref{thm:ABC}.
\end{proof}

%_______________________________%
\section{AB Systems and Important Special Cases} \label{AB}
In this section, we bound the performance of Ordinal Greedy on ABC Systems where $c = n$, effectively removing the component size constraint. We then discuss some common examples of maximization problems on AB Systems, including Max Spanning Tree, Max TSP, and Max Planar Subgraph. To improve our bound from $b+1$ we invoke the notion of sparsity.

\begin{definition} \label{def:sparsity} Sparsity\\
\emph{A graph $S$ is $d$-sparse if for all subgraphs $F \subseteq S$ containing $V(F)$ nodes and $E(F)$ edges $\frac{E(F)}{V(F)} \leq d$ and for any $\hat{d} < d$ there exists a subgraph $\hat{F} \subseteq S$ such that $\frac{E(\hat{F})}{V(\hat{F})} > \hat{d}$.}
\end{definition}

Suppose our attachment set $\mathcal{A}$ and degree limit $b$ imply that any feasible solution must be $d$-sparse. Note that this sparsity is implied by our constraints, and is not a separate constraint. Our main result in this section is that, for any graph collection which is guaranteed to be $d$-sparse, ordinal information is enough to produce good approximations. Specifically, we prove a bound of $d+1$ for such settings. Since the sparsity corresponds to an upper bound on average degree of the nodes, it is always true that $d \leq \frac{b}{2}$, and so when $c=n$, this immediately reduces the approximation factor from $b+1$ to $\frac{b}{2}+1$. Even for large $b$, however, there are many natural classes of graphs that are always sparse, including planar graphs, scale-free graphs, graphs of small arboricity or treewidth, and many others. As we discuss in the next section, this result allows us to provide extremely strong guarantees for many important problems.

\begin{theorem} \label{thm:AB}
For any ABC System where the components can be of any size and the constraints imply that any feasible solution must be $d$-sparse, the Ordinal Greedy algorithm always produces a solution within a factor of $\max\{2, (d+1)\}$ of the optimal solution, and this bound is tight.
\end{theorem}

\begin{proofsketch}
As with our proof of Theorem \ref{thm:ABC} for general ABC Systems, we only need to consider instances with weights $\{0,1\}$ due to Theorem \ref{thm:binary}. However, the charging schemes and proofs for ABC Systems and AB Systems differ significantly. To lower the approximation factor from $b+1$ to $\max\{2, d+1\}$, we have to be more selective about where we charge the edges of $OPT$. For simplicity, we first assign the edges of $OPT$ to their endpoints, before considering the total charge to all the nodes in any component. Since $OPT$ is $d$-sparse, the edges of $OPT$ can be assigned to their endpoints such that each node is assigned at most $d$ edges. We then take such an assignment and for each edge of $OPT$ eliminated due to $\mathcal{B}$ we change its assignment, if necessary, to the node which caused its elimination. Let $P$ be a component of the subgraph of $S$ containing only weight 1 edges, and suppose $p=|P|$. Then, similarly to the proof of Theorem \ref{thm:ABC}, we must show that this component will be charged at most $\max\{2, d+1\}w(P)$, but unlike before, components may be charged more than $(d+1)(p-1)$ if $w(P) > p-1$.

Now we consider two cases based on whether any node in a component was charged an edge of $OPT$ eliminated due to $\mathcal{B}$. If there such a node in a component, then it must be possible to distribute the charge on the nodes over the edges of the component directly so that no edge is charged more than $\max\{2,d+1\}$ times its weight. If there is no such node, then we show that at least one node in the component must be charged at most $p-1$ and the rest are charged at most $d(p-1)$, cumulatively providing a charge at most $(d+1)(p-1)$ which can then be distributed over the edges in the component. Once again, we provide a family of examples to show that this bound is tight.
\end{proofsketch}

\begin{proof}
In other words, we will show that $\alpha_{AB} = \max\{2, d+1\}$. For AB Systems, the fact that $c = n$ means no edge is eliminated due to $\mathcal{C}$. This leaves us with only 3 cases which can occur at the critical iteration of an edge.
%With just these cases, each edge of $OPT$ has its full weight charged to one of the components containing one of its endpoints.
We begin by creating an assignment of edges of $OPT$ to the nodes to construct our charging scheme. We use the same notation as in the proof of Theorem \ref{thm:ABC}.

As before, we know that the weight of a component $P$ is at least that of its maximum spanning tree, $w(P) \geq p - 1$. However, in certain cases, our strategy is different than the ABC Systems proof. We ensure that each component $P$ is charged at most $\max\{2,d+1\}w(P)$, but a component may be charged more than $\max\{2, d+1\}(p-1)$.

As in the proof of Theorem \ref{thm:ABC}, for each edge $(u,v)$ of $OPT$, we charge its weight between the components containing its endpoints. We need only to charge the edges of $OPT$ such that $w(u,v) = 1$, because the weight 0 edges do not contribute to the value of $w(OPT)$.

\begin{claim}\label{claim:flow}
There exists an assignment of edges of $OPT$ to their endpoints such that each node is assigned at most $d$ edges.
\end{claim}

\begin{proof} Let $N_d$ be a set of $dn$ nodes containing $d$ duplicates of each node in $\mathcal{N}$.
Let $O$ be a set of nodes where each node corresponds to an edge of $OPT$.
Construct a bipartite graph by building an edge from each node in $O$ to each of the $2d$ nodes in $N_d$ corresponding to the endpoints of the edge of $OPT$ it represents.

Consider any subset of nodes $R \subseteq O$ corresponding to a subset of edges in $OPT$. Let $V(R)$ be the set of endpoints in $\mathcal{N}$ of all edges of $OPT$ represented in $R$. Let $R'$ be the set of all edges of $OPT$ with both endpoints in $V(R)$, so that $R \subseteq R'$. Since $R'$ can have sparsity at most $d$, it follows that $|R'| \leq d \cdot |V(R)|$. Therefore, if we consider the $d \cdot |V(R)|$ nodes in $N_d$ corresponding to the duplicates of $V(R)$, which all have at least one edge to a node in $R$, we have that $|R| \leq |R'| \leq d \cdot |V(R)|$. By Hall's Condition, we can create a perfect matching between $O$ and $N_d$.

Take this perfect matching as described above, and assign each edge of $OPT$ to the endpoint in $\mathcal{N}$ corresponding to the node in $N_d$ to which they are matched. Since there $d$ duplicates of each node in $N_d$, each node in $\mathcal{N}$ can be assigned at most $d$ edges.

Note that while the proof above is written assuming that $d$ is an integer to duplicate the nodes, when $d$ is fractional the same result holds using a similar argument by duplicating the nodes in $O$.
\end{proof}

However, this assignment is not sufficient for our charging scheme. There may be too much assigned to small components $P$, and even to nodes which have no adjacent weight 1 edges in $S$ at all would be charged $d$, when they should not be charged at all. Therefore, we take this assignment and alter it to create our charging scheme so that every component $P$ is charged at most $(d+1)w(P)$.

Take the (possibly fractional) assignment of edges of $OPT$ to their endpoints from Claim \ref{claim:flow}, and for every edge of $OPT$ which was eliminated due to $\mathcal{B}$ change its assignment, if necessary, to be entirely to the endpoint which had degree $b$ at its critical iteration.

Given this new assignment, our charging scheme is now simple. For each node in $P$ charge the weight of the edges of $OPT$ assigned to it to $P$. We now show that each component $P$ has been charged at most $\max\{2,d+1\}w(P)$.

Let a \emph{b-node} be defined as a node with exactly $b$ adjacent edges of weight 1 in $S$, and therefore $b$ adjacent edges in $P$. By construction, only \emph{b-nodes} can receive a charge greater than $d$ and at most $b$, while all other nodes are charged at most $d$. If a node has degree $b$ in $S$ but is not a \emph{b-node} because one or more of its adjacent edges in $S$ has weight 0, then any edges of $OPT$ assigned to it which were eliminated due to $\mathcal{B}$ must have weight 0. The is because if any of the eliminated edges had weight 1, an edge of weight 0 could not be undominated at an iteration when the weight 1 edge is still adjacent to it in $E$. Therefore, a node which is not a \emph{b-node} can have at most $d$ edges of weight 1 assigned to it.

\textbf{Case 1)} $P$ contains at least one \emph{b-node}\\
Instead of showing that the total charge on $P$ is bounded by $\max\{2, d+1\}w(P)$, in this case it is simpler to think of the charge from the edges of $OPT$ as assigned to each particular node. The charge on $P$ is the cumulative charge on all of the nodes in $P$, which we distribute over the edges of $P$. For any component $P$, if the total charge to $P$ can be distributed over its edges so that each edge is charged at most $\max\{2, d+1\}$, then the total charge to $P$ is at most $\max\{2, d+1\} w(P)$. By summing over all $P\in \mathcal{P}$ we get $OPT(w) \leq \max\{2, d+1\}w(S)$.

All \emph{b-nodes} have charge at most $b$, which we can distribute so that each of its $b$ adjacent edges is charged at most 1. Now select one \emph{b-node} and consider the maximum spanning tree of $P$. For all nodes which are not \emph{b-nodes} distribute all $d$ of their charge to their adjacent edge on the path to the selected \emph{b-node} in this maximum spanning tree.

If any edge is between two \emph{b-nodes}, it is charged at most 2. If any edge has exactly one \emph{b-node} as an endpoint, it may be charged at most 1 from this endpoint and $d$ from the other for a total of $d+1$. For any edge between two nodes which are not \emph{b-nodes}, it is only charged from one of its endpoints, which has charge at most $d$. Since every edge of $P$ has at most $\max\{2, d+1\}$ charge, the total charge over $P$ is at most $\max\{2, d+1\}w(P)$.

\textbf{Case 2)} $P$ does not contain any \emph{b-nodes}\\
Let $u_0$ be the first node in $P$ to have an edge of weight 0 added adjacent to it in $S$ at some iteration of the Ordinal Greedy algorithm. If there is no node in P with an adjacent weight 0 edge in $S$, let $u_0$ be any node in $P$.

We show that the total charge to any component $P$ is at most $\max\{2, d+1\}(p-1)$ by showing that there are at most $d(p-1)$ edges of $OPT$ adjacent to the nodes of $P$, excluding those that have $u_0$ as their only endpoint in $P$. And there are at most $p-1$ edges of $OPT$ which have $u_0$ as their only endpoint in $P$.

Since there are no \emph{b-nodes}, all nodes in $P$ must be charged at most $d$. Clearly, $u_0$ could only be charged by edges which were eliminated due to $\mathcal{A}$ or are included in $P$. This means that for all edges of $OPT$ of weight 1 charged to $u_0$, their other endpoint must be in $P$. Therefore the total charge on $u_0$ is at most $min\{d, p-1\}$. Note that we cannot assume these edges have already been charged to $P$.

The sum of the total charge to the nodes of $P$ is at most $d (p-1) + min\{d, p-1\}$ which we can distribute over at least $(p-1)$ edges of $P$. We have that $\frac{d (p-1) + min\{d, p-1\}}{p-1} \leq d+1$. Our charging scheme ensures that each component $P$ is charged at most $(d+1)(p - 1) \leq (d+1) w(P) \leq \max\{2, d+1\}w(P)$.

Together with Case 1, we sum over $P \in \mathcal{P}$ and get $w(OPT) \leq \max\{2,d+1\}w(S)$, as desired.

\vskip 3pt
We now show that the above bound is tight.
Let $\mathcal{N} = \{u_1, ..., u_k, v_1, ..., v_k\}$.
Let $w(u_i, u_j) = 1$ for all $i \neq j$, $w(u_i, v_i) = 1+\epsilon$ for all $i \leq k$ for some infinitesimal $\epsilon$, and let all other edges have weight 0.
 Let $c = n$, $b = c-1$ and $\mathcal{A}$ = all planar subgraphs such that no cycle may contain an edge $(u_i, v_i)$ for any $i \leq k$.
This implies $d \leq \frac{6k-6}{2k} = \frac{3k-3}{k}$ because any planar graph on $n$ nodes has at most $3n-6$ edges.

For appropriate preference orderings $\sigma$, Ordinal Greedy selects the edges $(u_i,v_i)$ for all $i \leq k$ first.
Now edges $(v_i, v_j)$ with weight 0 may be undominated.
If any edge $(v_i, v_j)$ or $(u_i, u_j)$ is selected or eliminated the other must be eliminated at that iteration, but no ordinal algorithm can decide optimally between these edges because they are not adjacent.

Suppose Ordinal Greedy selects all edges $(v_i, v_j)$ where $j = i+1$, causing all edges $(u_i, u_j)$ with weight 1 to be eliminated due to $\mathcal{A}$, and proceeds by selecting additional $(v_i, v_j)$ edges until the solution is maximally planar. This yield $w(S) = k(1+\epsilon)$, whereas the optimal solution selects a maximal planar subgraph of edges $(u_i, u_j)$. The edges $(u_i, u_j)$ have a combined weight of $3k-6$ because any clique on $k$ nodes has a maximally planar triangulation of size $3k-6$. The edges $(u_i,v_i)$ have combined weight $k(1+\epsilon)$, for a total of $w(OPT) = 4k - 6 + \epsilon k$. Therefore as $\epsilon \rightarrow 0$, $\alpha_{AB} \geq \frac{w(OPT)}{w(S)}  = \frac{4k-6}{k} = \frac{3k-6}{k} +1$, which asymptotically approaches $\frac{3k-3}{k} + 1 = d+1$ as $k \rightarrow \infty$.
\end{proof}

The example for Max Weight Matching in Section \ref{sec:bMatching} provides a lower bound example where $d = \frac{1}{2}$, so $d < 1$ and $\alpha \geq 2$.
%_______________________________%

%_______________________________%
\subsection{Important Cases of AB Systems}

Theorem \ref{thm:AB} establishes that for AB Systems in which solutions are always sparse, ordinal algorithms don't perform much worse than ones which know the true underlying edge weights. While our result in the previous section is quite general, it is worth noting how it applies to many important problems which happen to be special cases of AB Systems. Since all tours and cycles are 1-sparse, and all planar graphs are at most 3-sparse, we immediately arrive at the following corollaries:

\begin{corollary} \label{cor:MST} Ordinal Greedy always computes a $2$-approximation for Maximum Weight Spanning Tree, and this bound is tight.
\end{corollary}

See Example \ref{example:MST} for example demonstrating the lower bound for Maximum Spanning Tree is tight.

\begin{corollary} \label{cor:TSP} Ordinal Greedy always computes a $2$-approximation for Maximum Traveling Salesperson, and this bound is tight.
\end{corollary}

\begin{proof} \label{proof:TSP_lower}
Let the set of agents be $\mathcal{N} = \{u_1, ..., u_k, v_1, ..., v_{k-3}\}$.
Let $w(u_i, u_{i+1}) = 1 + \epsilon$ for all $i \leq k-1$ for some infinitesimal $\epsilon$.
Let $w(v_i, u_{i+1}) = 1$ and $w(v_i, u_{i+2}) = 1$ for all $i \leq k-3$.
Let $w(u_1, u_{k-1}) = 1$ and $w(u_2, u_k) = 1$ and $w(u_1, u_k) = 1$.
Let all other edges have weight 0.

Suppose the Ordinal Greedy algorithm begins by selecting $(u_i, u_{i+1}) = 1+\epsilon$ for all $i < k$. This creates a path of length $k-1$ which causes all weight 1 edges in the graph to be eliminated due to $\mathcal{A}$ or $\mathcal{B}$. Ordinal Greedy may then proceed by selecting some set of weight 0 edges, such as $(u_1,v_1)$, $(u_k,v_{k-3})$, and $(v_i,v_{i+1})$ for $i < k-3$. Therefore, $w(S) = (k-1)(1+\epsilon)$. Meanwhile, the optimal solution is to select all of the weight 1 edges of the graph: $(v_i, u_{i+1})$ and $(v_i, u_{i+2})$ for all $i \leq k-3$ as well as $(u_1, u_{k-1})$, $(u_2, u_k)$, and $(u_1,u_k)$. This yields $w(OPT) = 2k - 3$. Therefore, $\alpha_{TSP} \geq \frac{w(OPT)}{w(S)} = \frac{2k-3}{(k-1)(1+\epsilon)}$ which asymptotically approaches 2 as $k \rightarrow \infty$ and $\epsilon \rightarrow 0$.
\end{proof}

\begin{corollary} \label{cor:planar} Ordinal Greedy always computes a $4$-approximation for Max Weight Planar Subgraph.
\end{corollary}

More generally, the same arguments can be applied to any problem where the goal is to find maximum-weight subgraphs with some excluded minor, finding maximum-weight graphs with small treewidth or arboricity, as well as a variety of other graph problems.
%_______________________________%

%_______________________________%
\section{$b$-Matching} \label{sec:bMatching}
For any ABC System where $c = n$ and $\mathcal{A}$ = all subgraphs of $G$, the only constraint is that each node must have degree at most $b$. This is equivalent to the well-known problem of Max Weight $b$-Matching. In this case, the approximation provided by Ordinal Greedy improves greatly over general AB Systems. In fact, it provides a strict $2$-approximation regardless of the value of $b$.

\begin{theorem} \label{thm:bMatching}
For any ABC System on graph $G$, where $c = n$ and $\mathcal{A}$ = all subgraphs of $G$, Ordinal Greedy always constructs a solution within a factor of 2 of the optimal solution. This bound is tight.
\end{theorem}

\begin{proof} Assume $\bar{w}: \mathcal{E} \rightarrow \{0,1\}$. We proceed in the same way as the proof of Theorem \ref{thm:AB}. For each weight 1 edge $(u,v) \in OPT$ eliminated due to $\mathcal{B}$, charge its full weight to its endpoint which has degree $b$. All $b$ adjacent edges to this node in the greedy solution must have weight 1 to have been undominated before the critical iteration of $(u,v)$, so it is a $\emph{b-node}$. Since no node can have more than $b$ adjacent edges in the optimal solution, this means that the total charge on any node is bounded by the number of edges of weight 1 adjacent to this node in the greedy solution. Therefore, all nodes can distribute their charge over their adjacent edges in $S$ such that no edge in the greedy solution is charged more than 2.

Consider the following minimal example to show that the factor of $2$ is tight. Suppose $b = 1$. Let $\mathcal{N} = \{u_1, u_2, v_1, v_2\}$. Let $w(u_1, u_2) = w(v_1, v_2) = w(u_1, v_1) = 1$. Let $w(u_2, v_2) = 0$. For suitable $\sigma$, Ordinal Greedy selects $(u_1,v_1)$ first because it is preferred by both $u_1$ and $v_1$, and the only other edge it can subsequently add is $(u_2, v_2)$. This yields $w(S) = 1$ and $w(OPT) = 2$.
\end{proof}

Note that $b = 1$ is the problem Max Weight Matching, and our result generalizes the results from \cite{anshelevich2016blind} and \cite{preis1999linear}.
%_______________________________%

%_______________________________%
\section{Conclusion and Further Directions}
In this paper we identify a large class of problems we call ABC Systems for which ordinal preference information is sufficient for algorithms to provide good approximations to optimal, even without access to cardinal utilities. Previous work has shown that if agent preferences form a metric space, approximations for TSP and matching can improve in expectation \cite{anshelevich2016blind}. It remains to be seen how Ordinal Greedy performs on ABC Systems in expectation and how much the approximation factors for general ABC or AB Systems improve when this metric assumption holds. Also along the lines of previous work, it would be interesting to investigate whether truthful ordinal algorithms for ABC and AB Systems can compete with our non-truthful algorithm, much as \cite{anshelevich2016truthful} did for the problems first approached in \cite{anshelevich2016blind}. Lastly, we have seen that all solutions produced by Ordinal Greedy are pairwise stable, but it is unknown for our problems whether all pairwise stable solutions produce a good approximation to optimum (although it is easily seen to be true for MST and TSP).

\setcounter{secnumdepth}{0}
\section{Acknolwedgments}
This work was partially supported by NSF award CCF-1527497.

%\bibliography{Full_Version_001}
%\bibliographystyle{plain}

\end{document}